\documentclass[twoside,leqno,twocolumn]{article}

\usepackage[letterpaper]{geometry}
\usepackage{ltexpprt}
\usepackage{amsmath,amsfonts,amssymb,mathrsfs} 
\usepackage{algorithmic}
\usepackage{graphicx}
\usepackage{textcomp}
\usepackage{xcolor}

\usepackage{caption} 
\usepackage{wrapfig}
\usepackage[procnumbered,linesnumbered,algo2e,algoruled,noend]{algorithm2e}
\usepackage{tikz}
\usetikzlibrary{tikzmark}
\usetikzlibrary{calc}
\usepackage{subcaption}
\SetKwRepeat{Do}{do}{while}
\usepackage[colorlinks=true,citecolor=black,linkcolor=black,urlcolor=blue]{hyperref}

\newcommand*{\lmulti}{\{\mskip-5mu\{}
\newcommand*{\rmulti}{\}\mskip-5mu\}}

\usepackage{array}
\usepackage{booktabs}

\newcommand{\V}{\mathcal{V}}
\newcommand{\E}{\mathcal{E}}
\newcommand{\T}{\mathcal{T}}
\newcommand{\N}{\mathcal{N}}

\definecolor{ForestGreen}{RGB}{34,139,34}
\definecolor{DarkGreen}{RGB}{14,79,14}
\definecolor{Blue}{RGB}{0,0,255}
\definecolor{Red}{RGB}{255,0,0}

 \SetCommentSty{mycommfont}

\newcommand{\matr}[1]{\mathbf{#1}}

\usepackage{cleveref} 
\Crefname{procedure}{Procedure}{Procedures}
\Crefname{conjecture}{Conjecture}{Conjectures}
\Crefname{definition}{Definition}{Definitions}
\Crefname{Definition}{Definition}{Definitions}
\Crefname{Theorem}{Theorem}{Theorems}
\Crefname{@theorem}{Theorem}{Theorems}
\newtheorem{conjecture}{Conjecture}[section]

\def\BibTeX{{\rm B\kern-.05em{\sc i\kern-.025em b}\kern-.08em
    T\kern-.1667em\lower.7ex\hbox{E}\kern-.125emX}}
\begin{document}

\title{Efficient Sampling of Temporal Networks with Preserved Causality Structure}
%
%
%
\author{Felix I. Stamm\thanks{RWTH Aachen University}
\and Mehdi Naima\thanks{Sorbonne Université, CNRS, LIP6, F-75005 Paris, France. This project has received financial support from the CNRS through the MITI interdisciplinary programs.} \and Michael T. Schaub\footnotemark[1]}
\date{}

%

\maketitle 

\begin{abstract}

In this paper, we extend the classical Color Refinement algorithm for static networks to temporal (undirected and directed) networks.
This enables us to design an algorithm to sample synthetic networks that preserves the $d$-hop neighborhood structure of a given temporal network.
The higher $d$ is chosen, the better the temporal neighborhood structure of the original network is preserved.
Specifically, we provide efficient algorithms that preserve time-respecting ("causal") paths in the networks up to length $d$, and scale to real-world network sizes.
We validate our approach theoretically (for Degree and Katz centrality) and experimentally (for edge persistence, causal triangles, and burstiness).
An experimental comparison shows that our method retains these key temporal characteristics more effectively than existing randomization methods.
 
\end{abstract}

\section{Introduction}
Graphs are a powerful framework for modeling complex systems across a multitude of fields.
However, in many systems, including communication networks, social interactions, financial systems, and biological processes the interaction patterns between entities are inherently dynamic, i.e., the graph structure changes over time.
To understand processes occurring on such networks, the temporal ordering of the interactions often plays a major role.
For example, the spread of pathogens or information can only occur within the graph along a temporal path whose edges occur sequentially in time, i.e. along a \emph{time-respecting walk}. 
Such temporal interaction sequences can be described by
temporal networks, which are essentially graphs with time-varying edge-sets.

To deepen our understanding of the underlying dynamical systems, it is crucial to analyze temporal networks
in ways that respect their core temporal structures.
In particular, randomizing these networks while maintaining time-respecting walks is essential for preserving the integrity of temporal causality and information flow, which are fundamental to the system's dynamics.
As highlighted in seminal works \cite{lambiotte2019,holme2012temporal,kempe}, time-respecting paths capture the ordered sequence of interactions that underpin real-world processes, i.e., a potential cause must temporally precede its result.
Failing to preserve this \emph{causality structure} would distort the temporal coherence, undermining the meaningful study of phenomena such as diffusion, communication, or influence within evolving systems.



To study the properties of real-world temporal networks, it is essential to have synthetic temporal network models that can preserve essential features of the observed temporal networks, while randomizing their structure otherwise. 
Samples drawn from such models can then be used as \textit{null models}. However, existing synthetic models for temporal networks, have so far mostly focused on preserving the degree distribution, as well as local structural elements like community structure, motifs or clustering coefficients. Models that aim to preserve the \textit{causal tree} (temporal neighborhood structure of the nodes) --- which is essential for many processes --- have so far received virtually no attention.

Generalizing ideas from~\cite{stamm2023}, in this paper we address this gap. 
We design an algorithm to efficiently generate synthetic network samples that (a) resemble the mesoscale path structure of a given temporal network up to a certain depth $d$, and are (b) maximally random otherwise. To this end, we first introduce the notion of temporal color refinement, which identifies nodes with their \textit{causal trees} that encodes the structure of the graph around node $v$, up to $d$ steps away.

We then introduce two graph rewiring mechanisms, one for directed and one for undirected temporal networks. 
These rewiring mechanisms guarantee that the refinement colors up to a specific depth are preserved (i.e., the temporal neighborhood structure remains invariant) even after rewiring the network.
Using a Markov Chain Monte Carlo (MCMC) sampling algorithm that employs these rewirings, we can thus obtain randomized synthetic graphs that are identically colored as a given input graph.

\noindent{\textbf{Contributions}}
We present a method that can \emph{efficiently generate synthetic networks} that mimic a given temporal network.
To that end, we first generalize the popular color refinement procedure to temporal graphs. We then show that it is possible to efficiently sample uniformly at random from all temporal graphs that are identically colored by temporal color refinement. 
Our sampling algorithm allows for efficient randomization of large temporal graphs since it is quasi-linear in all its parameters. 
We demonstrate empirically, that the sampled graphs are similar to the original graph with respect to many relevant graph measures.
We further proof that using the ultimate temporal refinement colors, the original and the sampled graph not only have the same degrees of all temporal nodes but also the same number of walks of any length starting from any temporal node which is known as temporal Katz centrality.

\begin{figure*}
    \centering
    \includegraphics[scale=0.95]{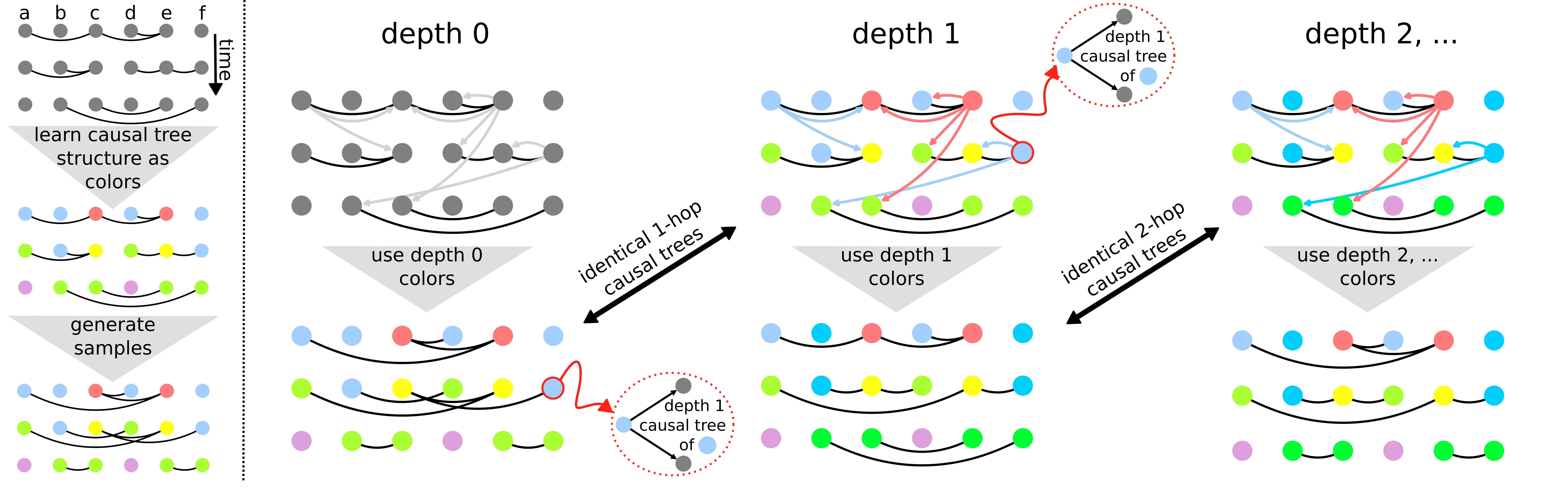}
    
    \caption{Left: High level overview of our procedure.
    Given an input temporal graph (here undirected, with six nodes and three times) we learn the causal tree structure (see \Cref{sec:app_def}) through temporal color refinement.
    This tree structure is then replicated in sampled temporal graphs.
    Right:
    Details of the proposed method.
    The top three graphs show the input graph colored at different depths $d$ of temporal color refinement.
    The refinement colors at depth $d$ distinguish nodes, if they have a different causal (successor) tree of depth $d$.
    The 1-hop causal successors of node $v$ at time $t$ are those nodes that may be reached by taking an edge from $v$ at any time later than at least $t$.
    As an example we indicate the causal successors of nodes $a$ and $e$ at time 1 and node $f$ at time 2 with small arrows.
    When using the depth $d$ colors as input to our procedure the ($d\,$+1)-hop causal trees of sampled graphs (bottom) are identical to the ($d\,$+1)-hop trees of the original graph.
    By generating temporal graphs that preserve the causal tree structure many properties of the original and sampled graphs are well preserved.
   As shown in this work, the temporal refinement colors and the sampling can be done in (log-)linear time and space making the proposed ``temporal NeSt'' procedure an efficient way to generate synthetic temporal graphs.
%
     }
    \label{fig:wl}
\end{figure*}

\subsection{Related work}
\subsection*{Temporal network models}

In this section we provide a brief overview of modeling approaches for temporal networks.
We separate three different branches of network models with slightly different goals.
Preservation driven models usually aim to preserve a certain network property such as degree or number of events per time while randomizing everything else as much as possible.
In latent variable models, explicit assumptions are made about how certain latent variables drive the time evolution of connections.
Synthetic graphs can then be generated by fixing the latent parameters (by fitting them to the data or otherwise) and sampling networks according to the underlying model assumptions.
Lastly, deep learning based models use machine learning to learn the time evolution of connections.

\noindent{\textbf{Preservation driven models}}
Gauvin et. al.~\cite{gauvinTempReview2022} consider generating networks by drawing uniformly at random from the set of temporal graphs that preserve a set of features of a given graph.
Other works, e.g. \cite{holmeEpidemiologicallyOptimalStatic2013}, fix the desired time-aggregated 
network and then assign times to the edges in this static network to obtain a temporal network.

\noindent{\textbf{Latent variable models}}
Zhang et. al.~\cite{zhangRandomGraphModels2017} generalize static network models, like the Erd\H{o}s R\'enyi, Chung-Lu, and the degree corrected stochastic blockmodel to the temporal setting by considering a continuous time evolution of the presence/absence of edges.
Mandjes et. al.~\cite{mandjesDynamicErdosRenyiGraphs2019} generalize this idea of edges appearing and disappearing in time at certain rates, by allowing the rate of changes to be themselves time dependent.
Liu and Sarıyüce~\cite{liuUsingMotifTransitions2023} generate temporal networks using motif transitions, i.e., larger temporal motifs are grown from smaller temporal motifs.
Porter et. al.~\cite{porterAnalyticalModelsMotifs2022} propose the TASBM which posits that the time between two subsequent edges connecting the same pair of nodes is exponentially distributed with a parameter that is proportional to the product of in-/out-activities of involved nodes.
Modeling temporal networks based on node activity is an approach that is also pursued by Perra et. al.~\cite{perraActivityDrivenModeling2012}.
More elaborate random graph models like the dynamic stochastic blockmodel~\cite{ghasemianDetectabilityThresholdsOptimal2016,xuDynamicStochasticBlockmodels2014} may be used when sufficient edges are present in each time slice.

\noindent{\textbf{Machine learning based models}}
Machine learning (deep learning) based approaches have been used to generate synthetic temporal graphs~\cite{guptaSurveyTemporalGraph2022}. But time and space requirements of proposed methods may hinder application to large datasets~\cite{liuUsingMotifTransitions2023}.
Specifically, to fit such models a large corpus of temporal networks is typically required. 
In contrast, we consider here a scenario in which we are given a single temporal network which we aim to randomize while preserving its temporal path structure.



\section{Formalism}\label{sec:form}
Commonly, temporal networks are used to model interaction between entities as a function of time.
While the duration of the interaction is relevant for some systems, in this work we only consider temporal networks for which the duration of the interaction is of no interest.
This is frequently the case when the duration of the interaction is short compared to the time until the next interaction.
The events of the temporal network can thus be well described by a triple consisting of the two actors involved and the time of the interaction.
Typically, there are only few events happening simultaneously.

Mathematically, we follow the formalism of \cite{grindrod2011}. 
We define a directed temporal graph $G$
as a triple $G = (\V, \E, \T)$ such that $\V$ is the set of vertices, $\E \subseteq \{ (u,v ) \,|\, u,v \in V, u \neq v \} \times \mathbb{N}$ is the set of temporal edges (transitions) and $\mathcal{T} = (t_1, t_2, \dots, t_T)$ is the sequence of strictly increasing timesteps with $\{t_i \in \T\} = \{t \, | \, ((u,v),t)\in \E\} $.
Throughout this work, we assume $t_1 \leq t_2 \leq \dots t_T$. 
The temporal edge $(( u,v), t) \in \E$ represents a directed temporal edge from $u$ to $v$ at time $t$.
We call a temporal graph undirected if every temporal interaction is bidirectional, i.e., $((u,v), t)\in \E \Leftrightarrow ((v,u), t)\in \E$.
We denote the number of vertices, edges and timestamps by $V := |\V|$, $E := |\E|$, $T := |\T|$, respectively.
We call $\V \times \T$ the set of temporal nodes.
In the top left corner of \Cref{fig:wl} a temporal graph with $\mathcal{V} = \{ a,b,c,d,e,f\}$ and $\T = ( 1,2,3)$ is shown.
\begin{Definition}[Temporal walk]\label{def:temp}
  Given a temporal graph $G = (\V, \E, \T)$, a temporal walk $W$ is a sequence of transitions $e \in \E^k$ with $k \in \mathbb{N}$, where $e = (e_1,\dots,e_k)$, with $e_i = ((u_i,v_i),t_i)$ such that for each $ 1 \leq i \leq k-1$ we have $v_i = u_{i+1}$  and $t_i \leq t_{i+1}$.
  The \textbf{length} of a temporal walk is its number of transitions. 
\end{Definition}
In the perspective of temporal graphs a \textit{causal} walk is a time-respecting walk. Therefore, causal and temporal walks refer the same concept.

\begin{Definition}[Successors of temporal node]
Given a temporal graph $G$. The (set of) successors of the temporal node $(v,t)$ is defined by:
\[ \mathcal S_G(v,t) = \{ (w,t') \; | \; ((v,w),t') \in \E, t \leq t' \}.\label{def::successor}\]
\end{Definition}

For instance, on the graph of Figure~\ref{fig:wl} the successors of node $e$ at time $1$ are $\mathcal{S}_G(e,1) = \{ (d,1),(c,1),(d,2),(c,3) \}$ (cf. the red arrows drawn for node $e$ in the column depth $d=1$ in \Cref{fig:wl}).

\begin{Definition}[Active and Sending nodes] We call a temporal node $(v,t)$ \emph{active} if there is some edge from or to node $v$ at time $t$.
    We call a temporal node $(v,t)$ \emph{sending} if there is a node $u$ such that $((v,u),t)\in \E$.
\end{Definition}
 As an example consider the input graph of \Cref{fig:wl}, the temporal nodes $(a,1)$ and $(a,2)$ are active, while $(a,3)$ is not active.
 We denote the static snapshot of $G$ at time $t$ by $G^{[t]} = (\V,\E^{[t]})$ with the set of edges present at $t$ being $\E^{[t]} = \{ (u,v)\, | \, ((u,v),t) \in \E\}$, and adjacency at time $t$ is denoted $A^{[t]}$.

Let us denote by $\mathcal W_{ij}(k)$ the set of all temporal walks from $i$ to $j$ of length $k$. Then, 
 \begin{Definition}[Temporal Katz centrality]
\label{def:katz}
    The Temporal Katz centrality of node $i$ is given by~\cite{grindrod2011}:
    \[ \Gamma_{\text{t-katz}}(i)  = \sum\limits_{j = 1}^n \sum\limits_{k \geq 0} \alpha^k\, |\mathcal W_{ij}(k)|,\]
    where $\alpha$ is a downweighting parameter.
\end{Definition}

Let $t_{i_1},t_{i_2},\dots,t_{i_\ell}$ be an increasing sequence of timestamps and defining
\[A =  A^{[t_{i_1}]} A^{[t_{i_2}]} \dots A^{[t_{i_\ell}]},\]
it can be shown~\cite{grindrod2011} that $(A)_{ij}$ counts the number of temporal walks of length $\ell$ from $i$ to $j$ that use exactly times $t_{i_1}, t_{i_2}, \dots t_{i_\ell}$. 
Therefore, to generalize Katz centrality to a temporal setting the authors of \cite{grindrod2011} sum all products of the form:
    $\alpha^{\ell}\,A^{[t_{i_1}]} A^{[t_{i_2}]} \dots A^{[t_{\ell}]}$. This results in defining
\begin{equation*}
    Q = (\mathit{I} - \alpha A^{[t_1]})^{-1} (\mathit{I} - \alpha A^{[t_2]})^{-1}\dots (\mathit{I} - \alpha A^{[t_T]})^{-1},
\end{equation*}
where $\mathit{I}$ is the identity matrix of order $V$.
The temporal Katz centrality can be written in terms of $Q$ as:
\begin{equation}\label{eq:katz}
    \Gamma_{\text{t-katz}} = Q \,\mathbb{1}_{V},
\end{equation}
where $\mathbb{1}_{V}$ is the all ones column vector of size $V$.
This is a ``send" centrality, i.e. it measures how well a node is able to reach other nodes in the network.
\section{Temporal Color Refinement: Efficiently learning neighborhood structure}\label{sec:algo}

Color refinement~\cite{wl} is an efficient procedure that is part of essentially all modern (static) graph isomorphism algorithms.
At its core, color refinement assigns each node of the graph a color and two nodes are assigned the same color if they are structurally similar. Other equivalence relationships between nodes in graphs exist such that \textit{regular equivalence}~\cite{regular_eq}.

In this section we generalize the color refinement procedure to temporal graphs and present an algorithm that can compute the temporal color refinement for all active temporal nodes of a graph in $O(d\, E\log E)$ where $d$ is the number of iterations needed until convergence.
This number is usually small for real world static graphs~\cite{stamm2023} or temporal graphs, as we will see in our experiments in Section~\ref{sec:experiments}.

In the temporal graphs literature, the static color refinement procedure was used when mining temporal motifs~\cite{kovanenTemporalMotifs2013}, but there it was applied to small, essentially static (sub-)graphs only.
In this work we extend the color refinement procedure to temporal graphs in order to learn the causal successor structure of the entire graph. Up to our knowledge, \emph{this is the first attempt to fully leverage the color refinement procedure to temporal graphs}.


\subsection{Temporal color refinement}

To generalise the color refinement from static to temporal graphs we consider a hash-based version of the color refinement procedure.
By replacing the normal neighborhoods of a node by their temporal (causal) successors (\cref{def::successor}) we arrive at the following recursion:
\begin{align}
     c^{(d+\!1)}_G\!(v,\!t)\! :=\! \mathrm{hash}\!\!\left(\!c^{(d)}_G\!(v,\!t), \lmulti c^{(d)}_G\!(w,\!\hat t)\, | \, (w,\!\hat t) \in \mathcal S_G(v,\!t) \rmulti\!\! \right)
     \label{eq:temp_WL}
\end{align}
which is a direct generalization of the equation for static graphs.
In the theoretical part of this paper we assume that the hash functions used are perfect, i.e. there are no hash collisions.
For the experiments the later mentioned practical hash function was used and no hash collisions were detected.
Mathematically, we may define the causal-tree (out-unraveling) of temporal node $(v,t)$ of depth $d$ for temporal graph $G$ as a directed tree $U_{\text{out}}^{(d)}(v,t) := (\V_G^d(v,t), \E_G^d(v,t))$.
To shorten the notation, in the following we write $\hat v:= (v,t)$ and $\hat u:= (u,t')$ to refer to the temporal node.
Consequently, we write $\V_G^d(\hat v):=\V_G^d(v,t)$ which applies similarly to the edges $\E_G^d$ and the successors $S_G$.

with nodes
$$
\V_G^d(\hat{v}) := \begin{cases}
    \{\hat{v}\} & d = 0,\\
    \bigcup\limits_{\hat{u} \in \mathcal S_G(\hat{v})}\{\hat {v} \frown x \mid x \in \V^{d-1}(\hat{u})\} & \text{otherwise.}
\end{cases}
$$
where $\frown$ denotes tuple concatenation.
We call the temporal node $\hat v$ the root of the causal-tree.
The edges of the causal-tree are
\begin{align*}
    \E^d(\hat{v}):= 
\bigcup_{\hat{u} \in \mathcal S_G(\hat{v})} \left\{(\hat v\frown a,\hat v\frown b) \mid (a,b) \in \E^{d-1}(\hat{u})\right\} \\ \cup \left\{(\hat{v}, \hat{v} \frown \hat{u})\right\}.
\end{align*}
Given a temporal graph $G = (\V,\E,\T)$ and a vector of initial colorings $c_G^{(0)}$ for the temporal nodes of $G$, this procedure assigns at each depth $d$ a hash (a ``color") to every temporal node.
We note that the aggregation described in \cref{eq:temp_WL} is happening backwards in time. That is, to obtain the colors of a temporal node at the next depth $d+1$ we are aggregating the colors of all nodes that are 1-hop reachable from the same node at future times.
As usual, this color assignment procedure refines the colors, i.e. 
\begin{align}
    c^{(d+1)}_G(v,t) = c^{(d+1)}_G(v,t) \Rightarrow c^{(d)}_G(v,t) = c^{(d)}_G(v,t).
    \label{eq:colors_refine}
\end{align}
Lets consider the implications of this equation.
Initially, all temporal nodes are in classes based on their initial colors (typically we start with a single class but the entire procedure works just as well when nodes are initially colored by the user). 
Before convergence is reached, each round of color refinement splits at least one of the classes. Hence, some temporal nodes that were in one class at depth $d$ are no longer in the same class at depth $d+1$.
The procedure converges when no color class is split from one round to the next.
The obtained color assignment of nodes to classes is then called the stable coloring of the temporal nodes.
See Figure~\ref{fig:wl} for an example of this coloring procedure.

\noindent{\textbf{Properties of temporal color refinement}}
Let us highlight some properties of temporal color refinement that are not present in color refinement for static graphs.
First, the successors of temporal nodes $(v,t_i)$ over increasing times fulfill
\begin{align}
    \mathcal S_G(v,t_1)\supseteq \mathcal S_G(v,t_2)\supseteq\dots \supseteq \mathcal S_G(v,t_T)
    \label{eq::nested _neighborhoods}
\end{align}
where $t_1\leq t_2\leq \dots \leq t_T$.
As the neighborhoods of earlier times contain potentially many nodes they are easily distinguished and nodes at early times are likely colored identical.
This is a consequence of the infinite look-ahead of nodes and may be mitigated, see the discussion in \cref{sec::conclusion}.
In particular if a temporal node $(v,t_i)$ is non-sending at time $t_i$ (i.e. there are no edges outgoing from $v$ at time $t_i$), then the successors are the same and the temporal nodes are colored identically if their initial color agrees:
\begin{align}
    c^{(d+1)}_G (v, t_i) = c^{(d+1)}_G (v, t_{i+1}) \Leftrightarrow c^{(0)}(v,t_i) = c^{(0)}(v,t_{i+1})\label{eq::refinement_completion}.
\end{align}
For uniform initial colors this implies that many consecutive temporal nodes have the same color. 
Unless stated otherwise we assume in the following that the initial colors $c^{(0)}$ are uniform.
Accordingly, in our statements, we will usually omit the explicit dependency on the (uniform) initial colors. However, most considerations presented still apply for non-uniform initial colors.
When computing the temporal color refinement of a network we are often mainly interested in the colors of active temporal nodes (temporal nodes $(v,t)$ having some activity at $t$). The colors of other temporal nodes may then be filled subsequently using \cref{eq::refinement_completion}.
As a consequence of \cref{eq::refinement_completion} there are at most $2E$ many colors needed to color the entire temporal graph, as there are at most $2E$ many active nodes.
As a consequence there are at most $2 E$ many rounds until stable colors are reached which is usually much less than the naive assumption that $V\,T$ many rounds are required.

\noindent\textbf{Temporal graphs with identical neighborhood structure}
We denote by $\N_G^d$ the set of all temporal graphs whose temporal nodes have the same colors after $d$ iterations  when starting from a uniform coloring, i.e. $c^{(0)} = \mathbb{1}_{n} = (1,1,\dots,1) \in \mathbb{R}^n$.
Note that $\N_G^1$ generalizes fixed degree graphs (configuration models) to a temporal setting.
For instance, when $G$ is a static graph ($T = 1$), we recover graphs with fixed degree sequence.

\subsection{Efficiently computing temporal refinement colors for active nodes}

Computing the refinement colors on \emph{static graphs} with $V$ nodes and $E$ edges by applying the equivalent of \cref{eq:temp_WL} for $d$ rounds will result in runtime $O(E \, d + V)$.
To converge to stable colors, $d$ may indeed become as large as $E$ --- although on real world graphs it is usually smaller than $E$.
With a better algorithm the colors can be computed in $O(E\log V+V)$ for any $d$~\cite{cardon1982,mckay1981}.
Both algorithms have space requirements of $O(V)$.

We may naively apply the above (faster) algorithm to compute the colors of all temporal nodes of a \emph{temporal network}. 
This would result in runtime complexity $O(E\, T\log(V\,T)+V\,T)$ and space requirements of $O(V\,T)$ as there are $V\, T$ temporal nodes in a temporal graph and each temporal edge $((u,v),t)$ affects all temporal nodes $(u,t')$ with $t'\leq t$ resulting in the factor $E\, T$.
Here, we propose Algorithm~\ref{alg:tcolor_refinemt} to compute the temporal color refinement of all \emph{active} nodes in $O(d\,E \log E)$ time and $O(V+E)$ space.
The colors of the active nodes are sufficient to obtain the colors for all temporal nodes as the remaining nodes can be colored using \cref{eq::refinement_completion}.
Computing the color refinement for only the active nodes already improves the runtime complexity of the naive algorithm to $O(E T\log(E)+E)$ and the space requirement is reduced to $O(E)$.
While these space requirements are acceptable, the runtime complexity still suffers from the $E T$ factor.
Our algorithm reduces this factor to $d\,E$, and obtains overall complexity of $O(d\, E \log E)$, where $d$ may be set by the user or is the number of rounds until the temporal color refinement algorithm converged. 
In theory $d$ may be as large as $E$, in practice $d$ is usually small.

\begin{algorithm2e}[t]
 \KwIn{\small{A Graph $G = (\V,\E,\T)$, number of rounds~$d$}}
 \KwOut{$\mathrm{colors}[(v,\!t)] \!=\! c^{(d)}(v,\!t)\, \forall\,$active nodes$\,(v,\!t)$}
 $h = \{ i \, : \, \mathrm{random\_hash}() \;| \;\;  i \in \{0,\dots,E\} \}$\\
 $ \mathrm{cs\_hash} = \{  (v,t)\, :\, 0 \;| \;\; \forall (v,t) \in \text{active nodes} \}$\\
 $ \mathrm{colors} \;\; = \;  \{  (v,t)\, :\, 0 \;| \;\; \forall (v,t) \in \text{active nodes} \}$\\
 $\mathrm{current} = 0,\, d' = 0 $\\
 \Do{$\mathrm{prev} \neq \mathrm{current} \text{ and } \;  d' < d$}{

    \tcp{Compute new hashes}
    \For{edge $((u, v), t) \in \E$}{
    $\mathrm{cs\_hash}[(u,t)] \text{ += } h[\mathrm{colors}[(v,t)]]$
    }
    \For{active $(v,t)$ in decreasing order of times}{
    \If{$t \neq n_v(t)$}{
        $\mathrm{cs\_hash}[(v,t)] \text{ += } \mathrm{cs\_hash}[(v,n_v(t))]$
    }
    }
    \tcp{Assign new colors}
    $\mathrm{last\_hash} = 0$\\
    $\mathrm{prev} = \mathrm{current};\;\; \mathrm{current} = -1$\\
    \For{$(v, t)$ in increasing order of $\mathrm{cs\_hash}$}{
        \If{$\mathrm{cs\_hash}[(v,t)]\neq \mathrm{last\_hash} $}{
        $\mathrm{last\_hash} = \mathrm{cs\_hash}[(v,t)]$\\
        $\mathrm{cs\_hash}[(v,t)] = 0; \;\; \mathrm{current} \text{ += } 1$
        }
        $\mathrm{colors}[(v,t)] = \mathrm{current} $
    }
    $d' \text{ += } 1$
 }
 
 \Return $\mathrm{colors}$
 \caption{Color refinement of active nodes}%
 \label{alg:tcolor_refinemt}
\end{algorithm2e}

The main observation used to move from quadratic to linear dependence on the number of edges is that we can compute a hash, that identifies (up to unwanted hash collisions) the color-multiset of the successors of all nodes in $O(E)$.
To achieve this we use an incremental multiset hash function~\cite{clarke2003incremental}.
The efficient use of the incremental multiset hash function is possible as the successors of temporal nodes $(v,t_i)$ are nested as per \Cref{eq::nested _neighborhoods}.


The main loop of \Cref{alg:tcolor_refinemt} has two phases.
In the first phase the multi-set hash values are calculated. In the second phase nodes with identical hash values are assigned identical colors.
To compute the hash values, we start with $E$ random numbers that serve as our hash values.
They should be chosen large enough and stored in array $h$ (in practice we use 64-bit hashes).
Think of array $h$ as assigning the hash $h[i]$ to a node with color $i$.
To compute the multi-set hash values according to \cref{eq:temp_WL} we use the additive multi-set hash introduced in \cite{clarke2003incremental}, where a multi-set hash is obtained by summing up hash values of individual elements (here the $h[i]$ values).
In our algorithm, we first sum up the hashes of successor nodes only within each time slice (lines 6-7).
The multi-set hash corresponding the full set of successors is obtained by cumulatively summing up the hashes backwards in time (lines 8-10, they are summed up for each node separately).
Throughout this hash calculation phase we make sure to properly deal with overflows as rollovers.

The $\mathrm{cs\_hash}$ values are then sorted in $O(E \log E)$ (line 13) and nodes are assigned new colors based on whether their hash values agree or not (lines 11-17).
Here we also reset the cs\_hash values to zeros to make them reusable in the next iteration.
This process of computing hash values and distinguishing nodes based on their colors needs to be repeated $d$ times or until stable refinement colors are reached (line 19).
Overall this results in $O(d \, E \log E)$ runtime.



\section{Generating synthetic temporal graphs with preserved neighborhood structure}\label{sec:theory}
In this section we show that it is possible to efficiently sample synthetic temporal graphs from the set $\N_G^d$ of all temporal graphs which are identically colored by the temporal color refinement procedure.
This means, we can efficiently generate synthetic networks which preserve the temporal neighborhood structure of a given temporal graph. 
To generate synthetic networks we perform Markov Chain Monte-Carlo (MCMC) sampling similar to the switch-chain of the configuration model~\cite{artzy2005generating,erdos2019mixing} or the NeSt model for static graphs~\cite{stamm2023}.
For the configuration model other sampling strategies like the curveball chain~\cite{carstens2015proof} or Las Vegas sampling methods~\cite{fosdick2018,greenhill2022} are known and might also be adaptable to our setting.
The key ingredient of our MCMC approach is a set of easy to perform rewirings. Multiple of these rewirings are then applied to each time slice graph $G^{[t]}$ resulting in another randomized graph. 
For a temporal graph $G = (\V,\E,\T)$ and $c$ some coloring of $G$ let us introduce the moves:
\begin{Definition}[Temporal swaps]
    If $G$ is undirected, a \textbf{swap} of two edges
     $(\{x,y\}, t)$, and $(\{r,s\}, t')$ consists in replacing them with $(\{x,s\}, t)$, $(\{r,y\}, t')$ if $c(x,t)=c(r,t')$, $c(y,t)=c(s,t')$, $x\neq s$, $r\neq y$ and $(\{x,s\}, t)$, $(\{r,y\}, t') \notin \E$.
    \label{def:swaps}
\end{Definition}

\begin{Definition}[Temporal tilts]
    If $G$ is directed, a \textbf{tilt} of an edge $((x,y), t)$ consists of replacing that edge with $((x,u),t')$ if $u \neq x$, $c(y,t)=c(u,t')$ and $((x,u), t') \notin \E$.
    \label{def:tilts}
\end{Definition}

A \textbf{swap} (respectively \textbf{tilt}) is called \textbf{in-time swap} (respectively \textbf{in-time tilt}) if $t = t'$.
Swaps and tilts are also referred to as \emph{rewirings}.
Interestingly, we find that only \emph{in-time rewirings} are needed to reach all temporal graphs having the same colors.
Given a temporal graph and fixed depth, Algorithm~\ref{alg:tnest} allows to sample (approximately) uniformly at random from all temporal graphs that are identically colored. The detailed rewiring procedures are presented in \cref{alg:undirnest} and \cref{alg:dirnest} that can be found in \Cref{sec:app_alg}. We assert the correctness of the Algorithm with the following theorems.

\begin{algorithm2e}[t]
 \KwIn{A Temporal Graph $G$, depth $d$, number of rewirings~$r$}
 \KwOut{Sample Graph $\tilde G$ approximately uniform from $\operatorname{t-}\mathcal N_G^d$}
 Use \Cref{alg:tcolor_refinemt} to obtain depth $d-1$ colors $c^{(d-1)}$ for all active temporal nodes\\
 \For{\text{all times} $t \in T$}{
    \tcc{rewire time slice $G^{[t]}$ using appropriate NeSt procedure}
    \uIf{$G$ \text{is undirected}}{
        $\tilde G^{[t]}$=undirNeSt($G^{[t]}$, $c^{(d-1)}(:, t)$, $r$)
    }
    \Else{
        $\tilde G^{[t]}$=dirNeSt($G^{[t]}$, $c^{(d-1)}$, $r$, $t$)
    }
    
 }
 return $\tilde G = \bigcup_{t \in T} \tilde G^{[t]}$%
 \caption{Sampling from temporal NeSt.
 }%
 \label{alg:tnest}
\end{algorithm2e}

\begin{theorem}[Output of Algorithm~\ref{alg:tnest}] \label{thm:rewiring_procedure}
Let $O(G,d)$ denote the set of all graphs reachable by applying multiple in-time color respecting rewirings, starting from $G$, and using depth $d$ colors obtained by applying the temporal WL procedure \cref{eq:temp_WL}.
This rewiring procedure, the ``temporal NeSt'' (short: t-NeSt) procedure, reaches exactly those graphs that are identically colored by the temporal WL procedure, i.e.
       \[ O(G,d) = \mathcal {N}^d_G.\]
\end{theorem}
A proof of \Cref{thm:rewiring_procedure} can be found in Appendix~\ref{sec:app_proof}.

\begin{theorem}[Unif. sampling of Algorithm~\ref{alg:tnest}]
Let us denote with $p_r(\tilde G, d, G)$ the probability, that the output of \cref{alg:tnest} is the graph $\tilde G$ using the number of rewirings $r$, number of rounds of the temporal WL algorithm $d$, and starting from $G$.
Then for sufficiently large $r$ the sampling is approximately uniformly at random from $\mathcal {N}^d_G$, i.e. for any graph $\tilde G \in \mathcal {N}^d_G$: 
       \[ \lim_{r\to \infty} p_r(\tilde G, d, G)=\frac{1}{|\mathcal {N}^d_G|}.\]
       \label{thm:uniform_sampling}
\end{theorem}
\begin{proof}
    The Markov chain used in \cref{alg:tnest} is undirected, aperiodic (there are self loops), strongly connected (by \Cref{thm:rewiring_procedure}) and has uniform out degree.
    The steady state of this Markov Chain is thus the uniform distribution.
\end{proof}

\begin{theorem}[Sampling time complexity]
       \label{thm:running time}
       The running time of Algorithm~\ref{alg:tnest} is $O(d\,E\log E + V+T+T\,r)$.
\end{theorem}
\begin{proof}
    The total runtime complexity of Algorithm~\ref{alg:tnest} without the running time of Algorithm~\ref{alg:tcolor_refinemt} is $O(V+T+E+T\,r)$. More details are in Appendix~\ref{sec:app_alg}.
\end{proof}

Theorem~\ref{thm:running time} allows for fast randomizations of large temporal graphs. As we shall see in Section~\ref{sec:experiments} the parameter $d$ is usually small for real-world temporal graphs and can thus be regarded as a constant and the sampling is quasi-linear in all other variables.
\subsection{Preserving temporal neighborhood structure preserves temporal Katz centrality}\label{sec:pres}
Identifying central nodes in networks is a common network analysis task.
In this section we consider common network centralities and show that the instant degree centrality and the Temporal Katz Centrality~\cite{grindrod2011} (see also \cref{def:katz}) is exactly preserved in samples of the temporal NeSt model.

\begin{table}[t]

    \centering\resizebox{0.5\textwidth}{!}{
    \begin{tabular}{l|c|c|rrr|r|rr|c}
    database & abb. & dire & nodes & times & edges & WL & $\#$t-NeSt($\infty$) & $\#$t-NeSt($1$) & ref  \\
    \hline
    ht09 & ht & 0 & 113 & 5.25K & 41.6K & 4 & 1 & 698 & \cite{ht}\\
    weaver & wea & 0 & 445 & 23 & 2.67K & 4 & 33 & 6.4K & \cite{wea}\\ 
    workplace & wp & 0 & 92 & 7.1K & 19.7K & 4 & 0 & 42 & \cite{wp}\\ 
    primate & pri & 0 & 25 & 19 & 2.68K & 2 & 17 & 584 & \cite{nr}\\ 
    racoon & rac & 0 & 24 & 52 & 3.99K & 3 & 0 & 706 & \cite{rac}\\ 
    dnc & dnc & 1 & 1.9K & 18.7K & 31.7K & 4 & 9.9M & 10.2M & \cite{dnc}\\ 
    fb-forum & fb & 1 & 899 & 33.5K & 33.7K & 4 & 5.58M & 6.2M & \cite{fb} \\
    email-eu2 & eu2 & 1 & 162 & 32.3K & 46K & 4 & 66K & 86.2K & \cite{eu}\\ 
    email-eu3 & eu3 & 1 & 89 & 8.91K & 12.1K & 4 & 16.1K & 22.3K & \cite{eu}\\ 
    talk\_eo & eo & 1 & 7.2K & 32.4K & 37K & 3 & 98.7M & 99.1M & \cite{eo} \\
    \end{tabular}
    }
    \caption{Statistics on the temporal graphs.
    From left to right: abb : abbreviated name, dire : indicated whether the graph is directed or undirected, nodes : the number of nodes, times : the number of time stamps, WL : the number of iterations until convergence of the color refinement procedure, $\#$t-NeSt($\infty$)/$\#$t-NeSt($1$) : the number of possible rewirings in the convergent depth/ with depth $1$ and ref : reference for the database provenance.}
    \label{tab:stats_db}
\end{table}


Let $G$ be a temporal graph, the instant degree centrality of a temporal node $d_G(v,t)$ is the number of outgoing edges from $v$ at time $t$.
Quite naturally the instant degree centrality of t-NeSt samples is identical.
Let $d\geq 0$, and $G' \in \N^d_G$ then for all temporal nodes $(v,t)$, $d_G(v,t) = d_{G'}(v,t)$.
This can be seen since \cref{alg:tnest} never changes the instant degree.

We similarly find for the Katz centrality:

\begin{theorem}\label{thm:centralities}
    Let $G_1 = (\V_1,\E_1,\T_1),G_2 = (\V_2,\E_2,\T_2) \in \N^\infty_G$, if two nodes $u,v \in \V_1 \cup \V_2$ are such that $c^{(\infty)}(u) = c^{(\infty)}(v)$, then:
    \[ \Gamma_{\text{t-katz}}(u) = \Gamma_{\text{t-katz}}(v).\] 
\end{theorem}
A proof of \Cref{thm:centralities} can be found in Appendix~\ref{sec:app_proof}.
The temporal Katz centrality defined using the matrix $Q$ can be seen as an ordered product of some operator over the time adjacency matrices of each time slice. In Katz centrality each term in the product is of the form $(\mathit{I} - \alpha A^{[t_i]})^{-1}$. Other forms of this product exist, for instance in~\cite{est_com} the author defines $\Gamma_{\mathrm{com}}$ that uses $\exp{\left(\beta\, A^{[t_i]} \right)}$ where $\beta$ is a fixed constant. The resulting centrality is called the network communicability or thermal Green’s function.
From observations in Section~\ref{sec:experiments}, Figure~\ref{fig:katz} we conjecture the following:
\begin{conjecture}\label{conj:com}
    Let $G_1 = (\V_1,\E_1,\T_1)$, $G_2 = (\V_2,\E_2,\T_2) \in \N^\infty_G$, if two nodes $u,v \in \V_1 \cup \V_2$ are such that $c^{(\infty)}(u) = c^{(\infty)}(v)$, then: $ \Gamma_{\mathrm{com}}(u) = \Gamma_{\mathrm{com}}(v).$
\end{conjecture}
\section{Empirical evaluation}\label{sec:experiments}

\begin{figure}[t]
    \centering
    \begin{subfigure}[b]{0.487\columnwidth}
    \includegraphics[width=1\columnwidth]{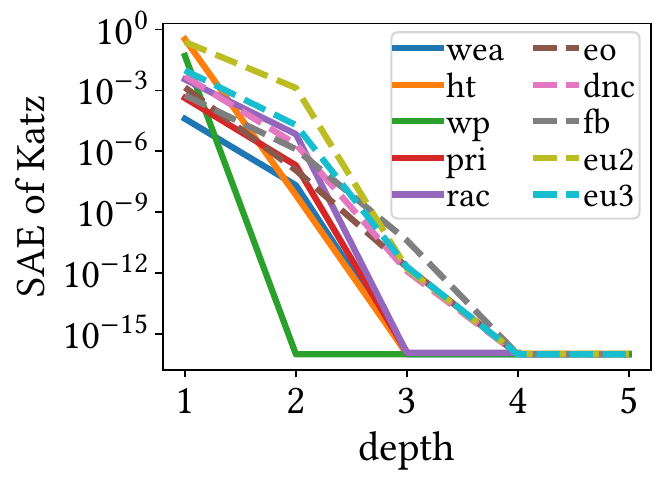}
    \caption{Katz centrality}
    \end{subfigure}
    \begin{subfigure}[b]{0.487\columnwidth}
\includegraphics[width=1\columnwidth]{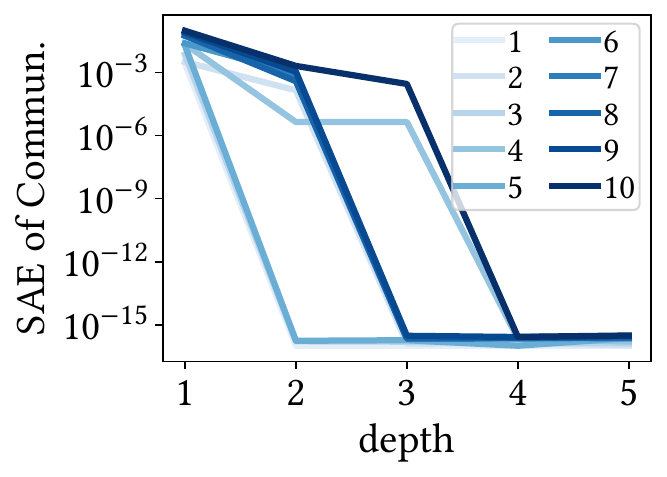}
    \caption{Communicability}
    \end{subfigure}
    \caption{Temporal network centralities as a function of depth.
    We show the sum absolute error (SAE) difference between the original and sampled graph.
    The left plot shows the SAE of Katz centrality $\Gamma_{\text{t-katz}}$ on real world graphs.
    The undirected graphs (solid lines) usually converge faster than the directed graphs (dashed lines).
    The right plot shows the SAE of Communicability centrality $\Gamma_{\mathrm{com}}$ for random graphs of varying density.
    The random graphs have 240 times, 80 nodes and each undirected temporal edges exist with probability $\tfrac{i}{80\cdot 240}$ where $i$ is indicated in the legend.
    Values are capped below by $10^{-16}$.
    We see that with increasing depth, centralities are better and better preserved.}
    \label{fig:katz}
\end{figure}

\begin{table*}[t]
    \centering
    \resizebox{1.0\textwidth}{!}{
    \begin{tabular}{l||l|llllll||l|llllll}
     data & origin. & t-NeSt($\infty$) & t-NeSt($1$) & DSS & RE & RT & RC & origin. & t-NeSt($\infty$) & t-NeSt($1$) & DSS & RE & RT & RC\\
    \hline
ht & 0.891 & $\mathbf{0.890(0)}$ & $0.890(0)$ & $0.23(3)$ & $0.005(0)$ & $0.055(3)$ & $0.004(1)$ & 157 & $\mathbf{157(0)}$ & $\mathbf{157(0)}$ & $77(3)$ & $77.80(5)$ & $160(0)$ & $28.90(3)$ \\
wea & 0 & $\mathbf{0}$ & $\mathbf{0}$ & $\mathbf{0}$ & $\mathbf{0.001(1)}$ & $-$ & $-$ & 0.201 & $\textbf{0.201(0)}$ & $0.184(1)$ & $0.124(1)$ & $0.012(4)$ & $-$ & $-$ \\
wp & 1.18 & $-$ & $\mathbf{1.18}$ & $0.80(2)$ & $0.002(1)$ & $0.046(2)$ & $0.005(1)$ & 35.8 & $-$ & $\mathbf{35.8}$ & $23(6)$ & $15.40(3)$ & $47.60(2)$ & $9.6(1)$ \\
pri & 0.084 & $\mathbf{0.084}$ & $\mathbf{0.084(1)}$ & $0.081(1)$ & $0.061(2)$ & $0.078(5)$ & $0.072(2)$ & 782 & $\mathbf{782(0)}$ & $783(0)$ & $774(2)$ & $766(4)$ & $641(21)$ & $476(14)$ \\
rac & 0.236 & $-$ & $\mathbf{0.232(1)}$ & $0.129(4)$ & $0.079(3)$ & $0.102(5)$ & $0.088(3)$ & 1151 & $-$ & $\mathbf{1150(1)}$ & $1141(7)$ & $1141(6)$ & $\mathbf{1151(0)}$ & $761(9)$ \\
dnc & 0.023 & $\mathbf{0.021(0)}$ & $0.017(0)$ & $0$ & $0(0)$ & $0.003(0)$ & $0.00(0)$ & 1.92 & $2.64(0)$ & $2.52(1)$ & $0$ & $0$ & $\mathbf{2.35(1)}$ & $0.046(1)$ \\
fb & 0.039 & $\mathbf{0.037(0)}$ & $0.019(0)$ & $0.00(0)$ & $0(0)$ & $0.001(0)$ & $0.00(0)$ & 0.001 & $\mathbf{0.001(0)}$ & $0.001(0)$ & $0.001(0)$ & $0.001(0)$ & $0.002(0)$ & $0.002(0)$ \\
eu2 & 0.061 & $\mathbf{0.052(0)}$ & $0.045(0)$ & $0.001(0)$ & $0.00(0)$ & $0.007(0)$ & $0.001(0)$ & 291 & $339(0)$ & $\mathbf{323(1)}$ & $1.35(3)$ & $1.92(2)$ & $383(1)$ & $27.40(2)$ \\
eu3 & 0.06 & $\mathbf{0.050(0)}$ & $0.038(1)$ & $0.002(0)$ & $0.00(0)$ & $0.007(1)$ & $0.001(0)$ & 9.04 & $\mathbf{9.16(2)}$ & $8.06(10)$ & $0.96(4)$ & $1.40(3)$ & $12.20(0)$ & $3.99(7)$ \\
eo & 0.15 & $\mathbf{0.155(0)}$ & $0.088(0)$ & $0$ & $0(0)$ & $0.001(0)$ & $0$ & 0.15 & $\mathbf{0.150(0)}$ & $0.152(1)$ & $0$ & $0$ & $0.220(1)$ & $0.001(0)$ \\

    \end{tabular}
    }
    \caption{Empirically determined values of the edge persistence $C(G)$ (left 7 columns) and triangles per temporal nodes $\Delta(G)$ (right 7 columns).
    We report the values for the original graph (origin.) and the six randomization models.
    The t-NeSt models are introduced here.
    DSS is degree constrained snapshot shuffling, RE stands for Randomized Temporal Edges~\cite{holme2012temporal,gauvinTempReview2022}, RT randomizes the times keeping the topology and multiplicities in the aggregated graph, and RC randomizes the number of context, keeping the aggregated graph topology.
    Each value is an average over $10$ randomizations and in parenthesis we provide the standard deviation on the last shown digit, if none is shown the values were preserved exactly.
    For both measures we highlight the values closest to the original graph.
    We omit valeswhen there are no possible rewirings.
    We find that the proposed t-NeSt procedure overall best preserves the edge persistency.
    For the number of triangles, we find that RT sometimes performs equally well and in one case even better than the t-NeSt procedures.}
    \label{tab:pers}
\end{table*}

In this section we will evaluate experimentally how t-NeSt behaves on real-world temporal graphs\footnote{For anonimity reasons the code will be made available upon publication}.
We therefore look at t-NeSt at different depths and compare it to existing temporal graph randomization schemes.
Definitions and further details can be found in Appendix~\ref{sec:app_exp}.
In our experiments we used four randomization schemes to compare with t-Nest.
We refer the interested reader to the review by Gauvin et al.~\cite{gauvinTempReview2022} that contains many temporal networks models.
\begin{enumerate}
\item \textbf{Randomized Edge (RE)}~\cite{li2017fundamental,holme2012temporal,holme_reach}.
RE only preserves the overall degree of a node $v$ (the number of temporal edges where $v$ is involved)
\item \textbf{Degree Snapshot Shuffling (DSS)}~\cite{sun,galimberti}. DSS preserves the time-slice degree of any node.
\item \textbf{Random Times (RT)}~\cite{kivela2012multiscale,holme2012temporal,posfai2014structural}. RT preserves the total number interactions between any two nodes but redistributes them over time.
\item \textbf{Randomized Contacts (RC)}~\cite{holme2012temporal,holme_reach}. RC preserves the topology of the aggregated network.
\end{enumerate}





We study the stability of the following temporal graph measures under randomization:
Edge persistence~\cite{nicosia2013graph} $C(G)$, Triangles~\cite{triangle} $\Delta(G)$, Causal triangles~\cite{triangle} $\Vec{\Delta}(G)$, and Burstiness~\cite{holme2012temporal} $B(G)$ defined in~\Cref{sec:app_exp}.
We further use Temporal Katz Centrality~\cite{grindrod2011} $\Gamma_{\text{t-katz}}(G)$ defined in~\Cref{eq:katz}, and Communicability Centrality~\cite{est_com} denoted $\Gamma_{\mathrm{com}}(G)$ defined in~\Cref{sec:pres}.

We consider directed and undirected temporal graphs from various domains. \Cref{tab:stats_db} gives general information about each dataset. We note that the number of possible rewirings in the directed case is much larger than in the undirected case where for some databases there are even no rewirings possible when $d=\infty$. This is due to the fact that undirected graphs are more constrained since edges go in both directions.

Our results on edge persistence and triangles are given in Table~\ref{tab:pers}. The results on burstiness can be found in \Cref{tab::B_active}, \Cref{tab::B_send} and \Cref{tab::B_receive} and finally, the results on causal triangles are given in \Cref{tab::cau_trian}. The results show that in most cases t-NeSt($\infty$) preserves better the values of the original graph than the other randomization schemes (the values of t-NeSt($\infty$) are closer to the original values). We also observe that t-NeSt($1$) values are the second closest to the original one.
This shows that the structure is better preserved with increasing depth. It is also interesting to notice that the results hold for both undirected and directed networks in a similar way. Therefore, the fact the directed networks can have a large number of rewirings (see \Cref{tab:stats_db}) compared to the undirected ones does not change the efficiency of t-Nest.

On the other hand, Figure~\ref{fig:katz} shows how centrality measures are well-preserved when randomizing using increasing depth. The results for Katz are in accordance with Theorem~\ref{thm:centralities} and lead us to the Conjecture~\ref{conj:com} seeing that the values of $\Gamma_{\mathrm{com}}$ become stable as $d$ increases.
Computing the stable colors for the presented graphs takes less than 1.5s on commodity hardware using no parallelism.

\section{Conclusion}
\label{sec::conclusion}
In this paper we introduced an efficient temporal graph randomization procedure that well-preserves the original graph structure. This was shown theoretically for some measures and experimentally for others. 

We have assumed that at any time $t$ we are aware of all future edges but one could also define the procedures for a \emph{finite look-ahead} $h$, i.e. an actor $v$ at time $t$ is only ``aware'' of future edges up to time $t+h$.
Many of the presented considerations may be applied to the finite horizon scenario as well while the complexity of presented algorithms remains unchanged. 
Further, according to our simulations in Figure~\ref{fig:katz} we expect that the network communicability to be preserved in t-NeSt samples. More generally, it would be if interest to characterize the measures that are preserved by the temporal color refinement procedure. Finally, it would be of interest to see how well this scheme generalizes to higher order structures such that \emph{hypergraphs} or \emph{simplicial complexes} and temporal versions of them.

 \bibliographystyle{plain}
  \bibliography{biblio}

\appendix
\section{Appendix}\label{sec:app}

\subsection{Additional Definitions}\label{sec:app_def}

\begin{Definition}
    Mathematically, we may define the causal-tree (out-unraveling) of temporal node $(v,t)$ of depth $d$ for temporal graph $G$ as a directed tree $U_{\text{out}}^{(d)}(v,t) := (\V_G^d(v,t), \E_G^d(v,t))$.
To shorten the notation, in the following we write $\hat v:= (v,t)$ and $\hat u:= (u,t')$ to refer to the temporal node.
Consequently, we write $\V_G^d(\hat v):=\V_G^d(v,t)$ which applies similarly to the edges $\E_G^d$ and the successors $S_G$.

with nodes
$$
\V_G^d(\hat{v}) := \begin{cases}
    \{\hat{v}\} & d = 0,\\
    \bigcup\limits_{\hat{u} \in \mathcal S_G(\hat{v})}\{\hat {v} \frown x \mid x \in \V^{d-1}(\hat{u})\} & \text{otherwise.}
\end{cases}
$$
where $\frown$ denotes tuple concatenation.
We call the temporal node $\hat v$ the root of the causal-tree.
The edges of the causal-tree are
\begin{align*}
    \E^d(\hat{v}):= 
\bigcup_{\hat{u} \in \mathcal S_G(\hat{v})} \left\{(\hat v\frown a,\hat v\frown b) \mid (a,b) \in \E^{d-1}(\hat{u})\right\} \\ \cup \left\{(\hat{v}, \hat{v} \frown \hat{u})\right\}.
\end{align*}
\end{Definition}

\subsection{Additions to Algorithms of Section~\ref{sec:algo} and Section~\ref{sec:theory}}\label{sec:app_alg}
With Algorithm~\ref{alg:tcolor_refinemt} we can compute the refinement colors of the active nodes using no more than $O(E)$ memory and $O(dE\log E)$ time.
We can then use \cref{eq::refinement_completion} to fill in the color of the remaining nodes and store that information in e.g. a two dimensional array of size $(T,V)$ which requires $O(VT)$ time and memory if we desire the temporal refinement colors of all nodes simultaneously and available in $O(1)$.

During the generation procedure it is not necessary to know all the colors of all the temporal nodes.
We will see that in the undirected case it is sufficient to know the refinement colors of only the active nodes.
In the directed case though, we need to know for each time $t$ the partition of all the temporal nodes $(v,t)$ into their respective color classes. However this knowledge is only needed per time slice and not the entire graph. We thus present time slice availability of the partition of nodes into colors can be achieved in $O(E+T)$ aggregated over all slices


The key idea to achieve availability of temporal colors (or the partition of nodes into color classes), is to realise that there are few times $t_i$ and $t_{i+1}$ when the color of temporal node $(v,t_i)$ is not equal to that of temporal node $(v,t_{i+1})$.
In particular this will only ever happen if node $(v,t_i)$ was sending.
Then node $(v,t_i)$ has more temporal successors than $(v,t_{i+1})$ and consequently they will have different colors.
An algorithm that efficiently provides the colors for each slice goes as follows.
Consider the slices increasing in times.
For the first slice we compute the colors for all nodes. For later times $t_i$, we only apply changes to the colors necessary by activity of nodes at time $t_{i-1}$.
This allows us to achieve a runtime of overall $O(V+E+T)$. Algorithm~\ref{alg:tcolors_available} summarises these steps.

\begin{procedure}[tb]
 \KwIn{Graph $G^{[t]}$, colors $c$, number of iterations $r$}
 \KwOut{Sample Graph $G'$ distributed as $\mathcal N_G(c)$}
 partition edges $uv$ of $G^{[t]}$ into subgraphs $g_{c_u, c_v}$\\
 let $\mathcal G$ be a list of all those subgraphs\\
    \For {$r$ steps}{
 Randomly choose a subgraph $g \in \mathcal G$\\

        \scalebox{0.8}[1.0]{choose two edges $u_1 v_1$ and $u_2 v_2$ unif. randomly from $E(g)$}\\
        \uIf{$|\{u_1, v_1, u_2, v_2\}| = 4 $}{ 
            \uIf {$u_1 v_2 \notin E(g) \text{ and } u_2 v_1 \notin E(g)$}{
                remove $u_1 v_1, u_2 v_2$\\
                ad\tikzmarknode{anch}{d} $u_1 v_2, u_2 v_1$
            }%
        }%

}

 \begin{tikzpicture}[remember picture,overlay,scale=0.5,every node/.style={scale=0.5}]%
               \node (alpha) at ($(anch.west)+(6.5, -.2)$)%
                 {%
                    \begin{tikzpicture}[scale=1.0, every node/.style={scale=1.0, inner sep=1pt, fill=black, circle, text=white, font=\huge}]
                \node[] (a) at (0,0) {$u_1$};
                \node[fill=blue] (b) at (0,1.5) {$v_1$};
                \node[] (c) at (1.5,0) {$u_2$};
                \node[fill=blue] (d) at (1.5,1.5) {$v_2$};
                \draw [-,line width=2pt] (a) to (b);
                \draw [-,line width=2pt] (c) to (d);
                    \end{tikzpicture}
                 };%
                \node (beta) at (alpha.east) [anchor=west,xshift=1cm]
                 {%
                    \begin{tikzpicture}[scale=1.0, every node/.style={scale=1.0, inner sep=1pt, fill=black, circle, text=white, font=\huge}]
                    \node[] (a) at (0,0) {$u_1$};
                    \node[fill=blue] (b) at (0,1.5) {$v_1$};
                    \node[] (c) at (1.5,0) {$u_2$};
                    \node[fill=blue] (d) at (1.5,1.5) {$v_2$};
                    \draw [-,line width=2pt] (a) to (d);
                    \draw [-,line width=2pt] (c) to (b);
                    \end{tikzpicture}
                 };%
               \draw [->,line width=1.0pt] (alpha)--(beta);%
\end{tikzpicture}%
 return $G' = \bigcup_{g\in \mathcal G} g$\\
\footnotesize{where $g_{c_u, c_v}$ corresponds to the graph having edges $(u,v)$ with $c{(u,t)}=c_u$ and $c{(v,t)}=c_v$, $E(g)$ is the set of edges of $g$}\\

 \caption{undirNeSt(), {\small Rewiring undirected graphs} }%
 \label{alg:undirnest}%
\end{procedure}

\begin{procedure}[tb]
 \KwIn{Graph $G$, colors $c$, iterations $r$, time $t$}
 \KwOut{Sample Graph $G'$ distributed as $\mathcal N_G(c)$}
 \uIf{$t$ is the first time}{Call \texttt{Initialize}($G$, $c$) of \Cref{alg:tcolors_available}} 
 c\_colors, u\_nodes = \texttt{Advance}($t$, $c$) of \Cref{alg:tcolors_available}\\
 Use c\_colors and u\_nodes to update node partitions\\
 \For{$r$ steps}{
Choose an edge $(u,v)$ unif. randomly from $\E^{[t]}$\\
        \scalebox{0.85}[1.0]{choose node $x$ unif. randomly from nodes colored like $v$}\\
        \uIf {$(u,x) \notin \E^{[t]}$ and $u \neq x$}{%
        remove edge $(u,v)$ from $\E^{[t]}$\\
        add \tikzmarknode{anch}{e}dge $(u,x)$ to $\E^{[t]}$
 \begin{tikzpicture}[remember picture,overlay,scale=0.5,every node/.style={scale=0.5}]%
               \node (alpha) at ($(anch.west)+(8.1, 0.0)$)%
                 {%
                    \begin{tikzpicture}[scale=1.0, every node/.style={scale=1.0, inner sep=4pt, fill=black, circle, text=white, font=\huge}]
                \node[] (a) at (0,0) {$u$};
                \node[fill=blue] (b) at (0,1.5) {$v$};
                \node[fill=blue] (d) at (1.5,1.5) {$x$};
                \draw [->,line width=2pt] (a) to (b);
                    \end{tikzpicture}
                 };%
                \node (beta) at (alpha.east) [anchor=west,xshift=1cm]
                 {%
                    \begin{tikzpicture}[scale=1.0, every node/.style={scale=1.0, inner sep=4pt, fill=black, circle, text=white, font=\huge}]
                    \node[] (a) at (0,0) {$u$};
                    \node[fill=blue] (b) at (0,1.5) {$v$};
                    \node[fill=blue] (d) at (1.5,1.5) {$x$};
                    \draw [->,line width=2pt] (a) to (d);
                    \end{tikzpicture}
                 };%
               \draw [->,line width=1.0pt] (alpha)--(beta);%
\end{tikzpicture}%
        }
    }
 
 return $G'=(\V,\E)$%

 \caption{dirNeSt(): {\small Rewiring directed graphs}}%
 \label{alg:dirnest}%
\end{procedure}



 \begin{algorithm2e}[tb]

\SetKwFunction{FAdvance}{Advance}
\SetKwFunction{FInit}{Initialize}
\SetKwProg{Fn}{Function}{:}{}
\Fn{\FInit{$G$, $\mathrm{colors}$}}{
 $h = \{ i  : \mathrm{random\_hash}(), \,  i \in \{0,\dots,|\mathrm{colors}|\} \}$\\
  $\mathrm{h\_colors} = dict()$\\
  $\mathrm{cur\_col} = \{ v:0\,,\, v \in \V \}$\\
  $\mathrm{cur\_h} = \{ v:0\,,\, v \in \V \}$\\
\For{ $(u,v, t) \in \E$}{
        $\mathrm{cur\_h}[u]$ += $ h[\mathrm{colors}[(v,t)]] $\\
     }
}
\BlankLine
\SetInd{0.3em}{0.5em}
\Fn{\FAdvance{t, $\mathrm{colors}$}}{
     $\mathrm{changed} = set()$\\
    \For{ $(u,v) \in \E^{[\mathrm{prev}(t)]}$}{
        $\mathrm{cur\_h}[u]$ -= $ h[\mathrm{colors}[(v,t)]] $\\
        $\mathrm{changed} = \mathrm{changed} \cup \{ u\}$\\
     }
     \For{$v \in \mathrm{changed}$ }{
        \If{$\mathrm{cur\_h}[v] \notin \mathrm{h\_colors}$}{
             $\mathrm{h\_colors}[cur\_h[v]] = |\mathrm{h\_colors}|$\\
        }
        $\mathrm{cur\_col}[v] = \mathrm{h\_colors}[\mathrm{cur\_h}[v]] $
     }
 }
 \Return $\mathrm{cur\_col}, \mathrm{changed}$\\
 \footnotesize{$\mathrm{prev}(t)$ is the preceding timestamp in the time sequence, $n_v(t)$ is the first timestamp $t'>t$ where node $v$ is active if none returns $t$.}
 \caption{{\small Efficient availability of temporal colors}}%
 \label{alg:tcolors_available}
\end{algorithm2e}
\subsection{Proofs of Section~\ref{sec:theory}}\label{sec:app_proof}

\subsubsection*{Causal completion of a temporal graph}
For a temporal Graph $G = (\V,\E,\T)$, the temporal graph restricted to times from $t_a$ to $t_b$ is
$G^{t_a:t_b}:=(\V, \{(u,v,t) | (u,v,t) \in \E, t_a\leq t \leq t_b\})$.
For temporal graphs we define the ``causal completion" $\tilde G = \text{CAUSAL(G)}:=(\tilde \V, \tilde \E)$ (a static directed graph) with
$\tilde \V = \{(u,t)\,|\, u\in \V, t\in \T\}$ and $\tilde \E = \{((v,t), (u, t'))| u,v \in \V,\, (u, t')\in \mathcal{S}_G(v,t)\}$.
This causal completion graph of $G$ is useful as the successors of a temporal node $(v,t)$ in $G$ correspond to the out neighbors (denoted $N^{\text{out}}_{\tilde G}(v,t)$) of $(v,t)$ in $\tilde G$. The adjacency matrix of the graph $\tilde G$ then has the form:
\begin{align}
  A_{\tilde G}=\begin{pmatrix}
A^{[t_1]} & A^{[t_2]}  & \cdots & A^{[t_T]}\\
0         & A^{[t_2]}  & \cdots & A^{[t_T]}\\
0         & 0         & \cdots & A^{[t_T]}\\
\vdots         & \vdots          & \ddots & \vdots\\
0         & 0          & \cdots & A^{[t_T]}\\
\end{pmatrix}  \label{eq:big_A}
\end{align}

This can be shown by induction on the number of timestamps $T:=|\T|$.
For $T=1$, we have that $\tilde G=\text{CAUSAL}(G)$
and $A_{\tilde G}=A^{[t_1]}_G$.
For the induction step let us assume that for any temporal graph $G=(\V,\E,\T)$ with $|\T| = N$ the adjacency matrix of the causal graph of $\tilde G=\text{CAUSAL}(G)$ has the structure introduced in \cref{eq:big_A}.
We will show the property for any temporal graph $G=(\V,\E,\T)$ with $|\T| = N+1$.
We see that we can apply the induction hypothesis on the subgraph $G^{t_2:t_{N+1}}$ since it has $N$ timestamps.
But the successors of node $v$ at time $1$ to nodes at any time $t$ is exactly given by the $v$-th row of the adjacency matrix of $G^{[t]}$.
This explains the shape of the first $V$ rows of $A_{\tilde G}$.
Lastly none of the nodes at times $t>1$ have successors to any node at time $t_1$. Thus the first $V$ columns are empty starting from row $V+1$.
The remaining entries of the adjacency matrix $\tilde G$ are obtained from the induction hypothesis.

\subsubsection*{Proof of Theorem~\ref{thm:rewiring_procedure}}
\vspace{-2mm}



Let $* \in \{\geq,>,=\}$, then:
\begin{align*}
 \mathcal S^{*}_G(v,t) &:= \{ (w,t') | (v,w,t') \in \E, t' * t \}.
\end{align*}
We also define the multisets of neighbor-colors:
\begin{align*}
    C^{d,*}_G(v,t) &:= \lmulti c^{d-1}(v,t) | (v,t) \in \mathcal S^{*}_G(v,t) \rmulti.
\end{align*}
Then we have that $\mathcal{S}_G^{\geq}(v,t) = \mathcal{S}_G(v,t)$. We notice that if the WL-colors are the same, then certainly the greater or equal multiset of neighbor-colors are also the same, i.e. $c_G^d(v,t)=c_{\tilde G}^d(u,t) \Rightarrow C^{d,\, \geq}_G(v,t) =  C^{d,\, \geq}_{\tilde G}(u,t)$. We now want to show: 
\[ O(G,d) = \mathcal{N}^d_G.\]

\textbf{Case $\mathbf \subseteq$:} The moves performed by the procedures 1 and 2 are a subset of all possible static NeSt~\cite{stamm2023} moves on the causal graph $\text{CAUSAL}(G)$.
Thus they preserve the colors.

\textbf{Case $\mathbf \supseteq$:}
We can extend the result of \cite{stamm2023} to the temporal case by induction on the total number of timestamps $T = |\T|$. Throughout the following proof let $d\geq 1$. \textbf{Base case} Case $T = 1$ is covered by the static NeSt procedure \cite{stamm2023}.\\
For the \textbf{induction step}, lets assume that
$O(G,d) \supseteq \mathcal{N}^d_G$
holds for all $G=(\V, \E, \T)$ with $T \leq N$. Let us start by assuming we have two graphs $G$ and $\tilde G$ identically colored at depth $d$, i.e. $G=(\V, \E, \T)$ and $\tilde G=(V, \E', \T) \in  \mathcal N^d_G$ with $|\T| = N+1.$ We can then apply the assumption of the induction to the temporal graphs consisting of all but the first time slice, i.e. we can mutually reach the temporal graph $\tilde G^{2:N+1}$ and the temporal graph $ G^{2:N+1}$ with the temporal NeSt procedure.
Consequently, starting from $\tilde G$ we can reach an intermediate graph $\hat G := \tilde G^{[1]} \cup G^{2:N+1}$ that only differs from $G$ on edges in the first timestamp.

In the following we show, that applying the rewirings (Definition~\ref{def:swaps} and Definition~\ref{def:tilts}) in the first timestamp transforms our intermediate graph $\hat G$ into the target graph $G$ which concludes the induction.
As the intermediate graph $\hat G$ was obtained from the starting graph $\tilde G$ by using color preserving switches, we know from $O(G,d) \subseteq \mathcal{N}^d_G$ that all three graphs are colored identical in all times. This means, that the greater or equal color multisets are the same for all nodes $v$, i.e.$\forall_{v\in \N}$
\begin{align*}
    \; C_{G}^{d,\, \geq}(v, 1) = C_{\hat G}^{d,\, \geq}(v, 1)
\end{align*}
But as $G$ and $\hat G$ have identical edges on all times except the first, i.e. $G^{2:N+1}=\hat G^{2:N+1}$, we have
\begin{align*}
    N^>_G(v,1) \overset{\text{Def.}}= N^\geq_G(v,2)= N^\geq_{\hat G}(v,2) \overset{\text{Def.}}= N^>_{\hat G}(v,1)
\end{align*}
If $\mathcal S^>_G(v,1) = \mathcal S^>_{\hat G}(v,1)$ then the color multisets  $C_G^{d,\, >}(v, 1)=C_{\hat G}^{d,\, >}(v, 1)$ are the same.
But now we have for all $v\in V$:


\begin{align*}
    C_G^{d, \geq}(v, 1) &= C_{\hat G}^{d, \geq}(v, 1)\; \implies
     C_G^{d,=}(v, 1) &= C_{\hat G}^{d, =}(v, 1)
\end{align*}
So the color multisets at time 1 are the same.

Coming back to the induction, we want to show that we can transform $\tilde G$ into $G$. From the induction hypothesis we know we can transform $\tilde G$ into $\hat G$, with $G$ and $\hat G$ agreeing on edges in all but the first time slice.
So what is left is to transform the first timestamp of our intermediate graph $\hat G^{[1]}$ into the first timestamp of our target graph $G^{[1]}$. But we just showed the intermediate and target graph are identically colored if we only consider neighbors at the same time slice (i.e. $\hat C^{d,\, =}(v, 1)=C^{d,\, =}(v, 1)$) which is exactly the prerequisite to applying the static NeSt procedure \cite{stamm2023} to the first timestamp. Thus overall starting from $\tilde G$ we can reach any other graph $G \in N^d_G$ using the algorithm which concludes the induction.
This shows $O(G,d) \supseteq N^d_G$.

\subsubsection*{Proof of Theorem~\ref{thm:centralities}}
In the following we show that the temporal-Katz centrality is a constant matrix multiplied with the Katz centrality of the causal completion graph $\tilde G$.
This is useful as the static NeSt moves preserve the static Katz centrality.
And the t-NeSt moves of the temporal graph $G$ are only a subset of all static NeSt moves possible on the corresponding causal completion graph $\tilde G$.
Thus as the temporal NeSt switches leave the static Katz centrality of the causal completion invariant which in turn leaves the temporal Katz centrality invariant.
All that remains to be shown is, that the temporal Katz centrality is indeed a constant times the non-temporal Katz centrality.

We start by considering finite terms $\Gamma_{\text{t-Katz}}^k$ for which we will show that they are identical to the temporal Katz centrality in the limit $k\to \infty$.
\begin{align*}
    \Gamma_{\text{t-Katz}}^k :&= (\matr{1}_{V\times V}\; \dots \; \matr{0}_{V\times V}) \Gamma_{\text{Katz}}^k(\tilde A)
\end{align*}
with $\Gamma_{\text{Katz}}^k(\tilde A):=\sum_{l=1}^k \alpha^l \tilde A^l  \mathbb{1}_{VT}$
where $\matr{1}_{V\times V}$ denotes the $V \times V$ identity matrix and $\mathbb{1}_{VT}$ the all ones column vector of size $VT$.
This agrees with the definition given in \cref{eq:katz} as we can rewrite the product $(\matr{1}_{V\times V}\; \dots \; \matr{0}_{V\times V})\tilde A^l  \mathbb{1}_{V\,T}$ as
\begin{align*}
    \Gamma_{\text{t-Katz}}^k = &\sum_{l=1}^k \alpha^l \sum_{j \in \T} (\tilde A ^l)_{\langle 1\rangle, \langle j\rangle}  \mathbb{1}_{V}\\
    \overset{\cref{lem:adjacency}}{=}&\sum_{l=1}^k \alpha^l \sum_{t_1\leq t_2 \dots \leq t_k} A^{[t_1]}A^{[t_2]}\dots A^{[t_k]}  \mathbb{1}_{V} \;\;=:Q_k \mathbb{1}_{V}
\end{align*}
For sufficiently small $\alpha$ we find that $\lim_{k\to \infty} Q_k = Q$ and thus $\Gamma_{\text{t-Katz}}^k\to\Gamma_{\text{t-Katz}}$.
Thus overall we find
$$\Gamma_{\text{t-Katz}}(G) = (\matr{1}_{V\times V}\; \dots \; \matr{0}_{V\times V}) \Gamma_{\text{Katz}}(\tilde G)$$

\begin{lemma}\label{lem:adjacency}
Let us denote with $A_{\langle i\rangle, \langle j\rangle}$ the $i,j$-th block of the matrix $A$ of size $V\,T \times V\,T$, i.e. the submatrix spanned by rows $(i-1)T +1$ through $i\, T $ and columns $(j-1)T +1$ through $j\,T $.
Then the powers of the adjacency matrix $\tilde A$ of causal completion graph $\tilde G$ of temporal graph $G$ have the following $V\times V$ blocks in the first block row
\begin{align*}
    (\tilde A ^k)_{\langle 1\rangle, \langle j\rangle} = \sum_{t_1\leq t_2 \dots \leq t_k=j} A^{[t_1]}A^{[t_2]}\dots A^{[t_k]}
\end{align*}
\end{lemma}
\begin{proof}
Through induction on $k$. The base case $k=1$ is apparent from the definition.
Lets assume that the proposition holds for $k$.
If we now take this matrix to the power of $k+1$, we find $\tilde A ^{k+1}_{\langle 1\rangle, \langle j\rangle} = (\sum_{i\leq j} A^{k}_{\langle 1\rangle, \langle i\rangle})A_{[j]}$.
But the sum of the first j entries $\sum_{i\leq j} A^{k}_{\langle 1\rangle, \langle i\rangle}$ are by our induction hypothesis all sorted products of $k$ matrices whose latest index does not exceed $j$, i.e. 
\begin{align}
    \sum_{i\leq j} A^{k}_{\langle 1\rangle, \langle i\rangle}=\sum_{t_1\leq t_2 \dots \leq t_k\leq j} A^{[t_1]}A^{[t_2]}\dots A^{[t_k]}
    \label{eq:A_helper1}
\end{align}
This means for the entries we have
\begin{align*}
    \tilde A ^{k+1}_{1 j}&=\left( \sum_{t_1\leq t_2 \dots \leq t_k\leq j} A^{[t_1]}A^{[t_2]}\dots A^{[t_k]} \right) A^{[j]}\\
    &=  \sum_{t_1\leq t_2 \dots \leq t_{k+1}=j} A^{[t_1]}A^{[t_2]}\dots A^{[t_{k+1}]}
\end{align*}
Which concludes the induction.
\end{proof}

\subsection{Additions to Experiments of Section~\ref{sec:experiments}}\label{sec:app_exp}
Note that all randomizing procedures preserve the number of edges per time.
DSS corresponds to t-NeSt($0$) which is the classical configuration model applied on the static graphs within each time slices (i.e \textbf{in-time swaps} if $G$ is undirected and \textbf{in-time tilts} if $G$ is directed).

RT and RC both preserve the topology (i.e. the edges) of the underlying static graph, but RT additionally preserves the multiplicity of the edges in the aggregated graph.

RE is the most relaxed procedure since it allows for any \textbf{swaps} or \textbf{tilts}, i.e. for undirected graphs, pick randomly two temporal edges $(a,b,t)$ and $(c,d,t')$ and with probability $1/2$ the edges become ($(a,c,t)$ and $(b,d,t')$) or ($(a,d,t)$ and $(b,c,t')$).
For directed graphs, pick up an edge $(a,b,t)$, choose uniformly any time $t'$, and any node $c$, replace $(a,b,t)$ with $(a,c,t')$.
For the t-NeSt and the RE procedure we made the same number of tentative rewirings.

The datasets used are available publicly at~\cite{kon},~\cite{nr} and~\cite{snap}.
\emph{All values shown here are averages over 10 samples} and the standard deviations are shown in brackets and apply to the last shown digit.

\begin{Definition}[Edge persistence~\cite{nicosia2013graph}]\label{def:pers}
For temporal graph $G$ the edge persistence is defined as:
    \[ C(G) = \frac{1}{E} \sum_{i \in \mathcal V}  \sum_{\ell = 1}^{T-1} \frac{\sum_{j} A^{[t_{\ell}]}_{ij} A^{[t_{\ell+1}]}_{ij} }{\sqrt{ (\sum_{j}A^{[t_{\ell}]}_{ij} ) (\sum_{j}A^{[t_{\ell+1}]}_{ij} ) }}, \]
where $E$ is the number of temporal edges. In~\cite{nicosia2013graph}, the authors normalized by $(V\, (T-1))$ while we normalized by $E$ which is more suitable for very sparse temporal graphs.  
\end{Definition}

As the name edge persistence already indicates it measures how many edges persist from one temporal slice to the next one.
Naturally for sparse graphs where there are very few edges per time and many nodes, one would expect a low edge persistence while for graphs with many edges per slice the persistence could be expected to be higher.

\begin{Definition}[Triangles~\cite{triangle}]\label{def:tri}
Let $G$ be a temporal graph, then a triangle corresponds to a set $ s = \{ e_1,e_2,e_3\}$ where $e_1 = ((a,b),t_1), e_2 = ((b,c),t_2)$ and $e_3 = ((c,a),t_3)$ for $a,b,c \in V$ and $t_1,t_2,t_3 \in \T$. This corresponds to shapes $\mathcal{T}_4$ or $\mathcal{T}_8$ of~\cite{triangle}. Then we define $\Delta(G) = { \# \mathrm{triangles}}/(VT)$.  
\end{Definition}

Triangles and are well known to be a very relevant characteristic of static (social) networks, as they are a result of a community structure or potentially of a kind of triadic closure effect.
Similarly, for temporal networks a large number of triangles may hint at a hidden community structure or at processes hat favor triangles, such as emails going in short circles or spacial proximity resulting accidental triangular meetups over time.

\begin{Definition}[Causal Triangles~\cite{triangle}]\label{def:cau_tri}
Let $G$ be a temporal graph, then a causal triangle corresponds to a set $ s = \{ e_1,e_2,e_3\}$ where $e_1 = ((a,b),t_1), e_2 = ((b,c),t_2)$ and $e_3 = ((c,a),t_3)$ for $a,b,c \in V$ and $t_1,t_2,t_3 \in \T$ such that $t_1 < t_2 < t_3$.
This corresponds roughly to the shape $\mathcal{T}_8$ of~\cite{triangle}.
Then we define $\Vec{\Delta}(G) = { \# \mathrm{causal triangles}}/(VT)$.  
\end{Definition}

While the total number of triangles is an interesting and frequently looked at measure, it may be computed from the aggregated graph with its multiplicities.
But the temporal ordering of events is lost in the aggregated graph, it thus makes sense for temporal networks to not only look at temporal triangles but also at causal triangles.

\begin{table}[th]
\resizebox{1.0\columnwidth}{!}{
\begin{tabular}{l|l|llllll}
\toprule
abbr & origin. & t-Nest($\infty$) & t-Nest($1$) & DSS & RE & RT & RC \\
\midrule
ht & 73.2 & $\mathbf{73.20(0)}$ & $73.20(0)$ & $38(2)$ & $38.60(3)$ & $80.00(4)$ & $14.40(2)$ \\
wea & 0 & $\mathbf{0}$ & $\mathbf{0}$ & $\mathbf{0}$ & $0.003(1)$ & $0.088(3)$ & $0.088(3)$ \\
wp & 19.2 & $\mathbf{19.2}$ & $\mathbf{19.2}$ & $12(3)$ & $7.60(1)$ & $23.80(3)$ & $4.80(6)$ \\
pri & 291 & $-$ & $\mathbf{291(0)}$ & $294(2)$ & $294(2)$ & $272(8)$ & $202(6)$ \\
rac & 502 & $\mathbf{502(0)}$ & $\mathbf{502(0)}$ & $523(3)$ & $525(3)$ & $543(2)$ & $358(4)$ \\
dnc & 0.972 & $1.31(0)$ & $1.25(1)$ & $0$ & $0$ & $\mathbf{1.18(0)}$ & $0.023(1)$ \\
fb & 0.001 & $\mathbf{0.001(0)}$ & $0.001(0)$ & $0.00(0)$ & $0.00(0)$ & $0.001(0)$ & $0.001(0)$ \\
eu2 & 147 & $171(0)$ & $\mathbf{163(0)}$ & $0.67(2)$ & $0.96(1)$ & $191(1)$ & $13.70(1)$ \\
eu3 & 4.6 & $\mathbf{4.69(1)}$ & $4.12(5)$ & $0.48(2)$ & $0.70(2)$ & $6.10(2)$ & $1.99(3)$ \\
eo & 0.076 & $\mathbf{0.076(0)}$ & $0.078(1)$ & $0$ & $0$ & $0.110(0)$ & $0.001(0)$ \\
\bottomrule
\end{tabular}
}
\caption{Empirically determined values for the normalized number of causal triangles $\Vec{\Delta}$.
We find that in many networks the number of causal triangles is well preserved by the temporal NeSt procedure.
Only on the dnc network the Random Times (RT) procedure better preserves the number of causal triangles.
}
\label{tab::cau_trian}
\end{table}

\begin{Definition}[Burstiness~\cite{holme2012temporal}]\label{def:burs}
Let the $i$-th time that nodes $u$ is active as $\T^{\text{active}}_i(u)$, and the number of times node $u$ was active as $T(u)$.
Similar in the directed case we distinguish the $i$-th time a node is sending/receiving as $\T^{\text{send}}_i(u)$ / $\T^{\text{receive}}_i(u)$ and similarly the number of times a node is sending/receiving as $T^{\text{send}}(u)$ or $T^{\text{receive}}(u)$ respectively.

Then for a graph $G$ the multi-set of inter arrival times is defined as $\tau^X_G:= \bigcup_{u \in \V(G)}\lmulti \T^X_{i+1}(u)-\T^X_i(u) | 1\leq i < T^{X}(u)\rmulti$ (the additive union is used) for superscript $X\in \{\text{active},\text{send}, \text{receive}\}$.
Then the active/send/receive burstiness is given as
    $$B^X(G) = \dfrac{\sigma_{\tau^X_G} / m_{\tau^X_G} -1}{\sigma_{\tau^X_G} / m_{\tau^X_G} +1} = \dfrac{\sigma_{\tau^X_G}-m_{\tau^X_G}}{\sigma_{\tau^X_G} + m_{\tau^X_G}},$$
    where $m_{\tau^X_G}$ is the mean of inter contact times $\tau^X_G$ and $\sigma_{\tau^X_G}$ its standard deviation.
\end{Definition}

The burstiness allows to asses the uniformity of the temporal process both as a function of the nodes and as a function of time itself.
If the temporal activity for each node approximately follows a Poisson process with parameter $\lambda$ uniform across all nodes and uniform across time then the burstiness is approximately zero $B\approx 0$.
If the average inter contact time is large compared to the standard deviation, then the process is extremely uniform and the burstiness is close to minus one.
All the networks seen here show quite some variation in their inter contact times as the burstiness is usually at least 0.5 which implies the standard deviation is at least three times larger than the mean.

\begin{table}[th]

\resizebox{1.0\columnwidth}{!}{
\begin{tabular}{l|l|llllll}
\toprule
abbr & origin. & t-Nest($\infty$) & t-Nest($1$) & DSS & RE & RT & RC \\
\midrule
ht & 0.795 & $\mathbf{0.795}$ & $\mathbf{0.795}$ & $\mathbf{0.795}$ & $0.717(2)$ & $0.701(3)$ & $0.690(3)$ \\
wea & 0.823 & $\mathbf{0.823}$ & $\mathbf{0.823}$ & $\mathbf{0.823}$ & $0.048(9)$ & $0.031(8)$ & $0.031(8)$ \\
wp & 0.733 & $\mathbf{0.733}$ & $\mathbf{0.733}$ & $\mathbf{0.733}$ & $0.631(2)$ & $0.629(2)$ & $0.611(1)$ \\
pri & 0.571 & $\mathbf{0.571}$ & $\mathbf{0.571}$ & $\mathbf{0.571}$ & $0.435(6)$ & $0.84(10)$ & $\mathbf{0.5(2)}$ \\
rac & 0.458 & $-$ & $\mathbf{0.458}$ & $\mathbf{0.458}$ & $0.36(2)$ & $0.33(1)$ & $0.252(5)$ \\
dnc & 0.855 & $0.867(1)$ & $\mathbf{0.848(2)}$ & $0.728(1)$ & $0.696(3)$ & $0.779(6)$ & $0.723(6)$ \\
fb & 0.613 & $\mathbf{0.612(1)}$ & $\mathbf{0.612(1)}$ & $0.415(2)$ & $0.380(3)$ & $0.528(5)$ & $0.463(6)$ \\
eu2 & 0.811 & $\mathbf{0.812(1)}$ & $0.815(2)$ & $0.780(1)$ & $0.780(0)$ & $0.798(2)$ & $0.780(1)$ \\
eu3 & 0.735 & $\mathbf{0.736(0)}$ & $\mathbf{0.735(2)}$ & $0.698(1)$ & $0.696(1)$ & $0.707(3)$ & $0.702(2)$ \\
eo & 0.726 & $\mathbf{0.748(1)}$ & $0.617(1)$ & $0.274(1)$ & $0.141(1)$ & $0.511(2)$ & $0.360(1)$ \\

\bottomrule
\end{tabular}
}
\caption{Empirically determined values for the active burstiness $B^{\text{active}}$.
}
\label{tab::B_active}
\end{table}

The result for the active burstiness is shown in \Cref{tab::B_active}.
Naturally in the undirected case the active burstiness $B^{\text{active}}$ is preserved exactly by the t-NeSt procedure and as DSS is just t-NeSt(0), DSS preserves it exactly as well.
In the directed case it is not guaranteed that the active burstiness is preserved well.
Still, the NeSt procedure better preserve the burstiness than other randomization methods.

\begin{table}[th]
\resizebox{1.0\columnwidth}{!}{
\begin{tabular}{llllllll}
\toprule
abbr & origin. & t-Nest($\infty$) & t-Nest($1$) & DSS & RE & RT & RC \\
\midrule
dnc & 0.862 & $\mathbf{0.862}$ & $\mathbf{0.862}$ & $\mathbf{0.862}$ & $0.621(4)$ & $0.759(4)$ & $0.708(6)$ \\
fb & 0.631 &  $\mathbf{0.631}$ &  $\mathbf{0.631}$ & $\mathbf{0.631}$ & $0.270(3)$ & $0.510(5)$ & $0.423(6)$ \\
eu2 & 0.758 & $\mathbf{0.758}$ & $\mathbf{0.758}$ & $\mathbf{0.758}$ & $0.699(1)$ & $0.739(3)$ & $0.726(2)$ \\
eu3 & 0.666 & $\mathbf{0.666}$ & $\mathbf{0.666}$ & $\mathbf{0.666}$ & $0.585(3)$ & $0.628(3)$ & $0.611(4)$ \\
eo & 0.746 & $\mathbf{0.746}$ & $\mathbf{0.746}$ & $\mathbf{0.746}$ & $0(2)$ & $0.596(5)$ & $0.568(3)$ \\

\bottomrule
\end{tabular}
}
\caption{Empirically determined values for the send burstiness $B^{\text{send}}$.
}
\label{tab::B_send}
\end{table}

The empirical results for the send burstiness is shown in \Cref{tab::B_send}.
As the out-degree of each temporal node in each slice is preserved exactly by the NeSt and DSS procedure thus the send burstiness is preserved exactly them as well.
The margin by which the send burstiness is better preserved by the NeSt networks is usually at least 0.1 except for the eu2 and eu3 network for which the RT method also preserves the bustiness relatively well.

\begin{table}[th]
\resizebox{1.0\columnwidth}{!}{
\begin{tabular}{l|l|llllll}
\toprule
abbr & origin. & t-Nest($\infty$) & t-Nest($1$) & DSS & RE & RT & RC \\
\midrule
dnc & 0.821 & $0.844(0)$ & $\mathbf{0.820(3)}$ & $0.637(2)$ & $0.755(6)$ & $0.752(8)$ & $0.685(7)$ \\
fb & 0.537 & $\mathbf{0.532(4)}$ & $\mathbf{0.547(5)}$ & $0.270(4)$ & $0.494(6)$ & $0.497(6)$ & $0.449(5)$ \\
eu2 & 0.741 & $\mathbf{0.742(1)}$ & $0.746(3)$ & $0.698(1)$ & $0.727(3)$ & $0.727(3)$ & $0.705(2)$ \\
eu3 & 0.646 & $\mathbf{0.647(1)}$ & $\mathbf{0.644(4)}$ & $0.591(2)$ & $0.604(4)$ & $0.608(3)$ & $0.596(5)$ \\
eo & 0.627 & $\mathbf{0.664(1)}$ & $0.495(1)$ & $0(2)$ & $0.370(2)$ & $0.369(2)$ & $0.144(2)$ \\

\bottomrule
\end{tabular}
}
\caption{Empirically determined values for the receive burstiness $B^{\text{receive}}$.}
\label{tab::B_receive}
\end{table}

The values for the receive burstiness are shown in \Cref{tab::B_receive}.
Unlike the send burstiness is this not guaranteed to be well preserved by the t-NeSt kind of models.
We find that the NeSt models preserve the receive burstiness in four out of five networks up to the first two digits.
For the eu2 network the other randomization schemes result in a receive burstiness which is only 3\% away from the original value, while in the other cases the relative difference is about 10\%.



\end{document}